\documentclass[a4paper]{article}
\usepackage{latexsym,amsmath,amssymb,enumerate,mathrsfs,amsthm}
\usepackage{graphicx,color}
\usepackage{cite}

\newcommand{\alg}[1]{\mathfrak{#1}}

\newcommand{\mc}[1]{\mathcal{#1}}
\newcommand{\cat}[1]{\textrm{\bfseries #1}}
\newcommand{\qd}[1]{\mathcal{D}(#1)}
\newcommand{\qdz}{\qd{\mathbb{Z}_2}}
\newcommand{\End}{\operatorname{End}}

\newcommand{\id}{\operatorname{id}}
\newcommand{\str}{\operatorname{star}}
\newcommand{\plaq}{\operatorname{plaq}}
\newcommand{\Ad}{\operatorname{Ad}}
\newcommand{\spec}{\operatorname{Sp}}
\newcommand{\dist}{\operatorname{dist}}
\newcommand{\supp}{\operatorname{supp}}
\newcommand{\norm}[1]{\lVert#1\rVert}
\newcommand{\Rep}{\cat{Rep}}
\newcommand{\bonds}{{\bf B}}
\newcommand{\set}[1]{\left\{ #1 \right\}}
\DeclareMathOperator*{\wlim}{w-lim}

\theoremstyle{plain}
\newtheorem{theorem}{Theorem}
\numberwithin{theorem}{section}
\newtheorem{lemma}[theorem]{Lemma}
\newtheorem{proposition}[theorem]{Proposition}
\newtheorem{corollary}[theorem]{Corollary}

\theoremstyle{definition}
\newtheorem{definition}[theorem]{Definition}

\theoremstyle{remark}
\newtheorem{remark}[theorem]{Remark}

\numberwithin{equation}{section}

\title{Localized endomorphisms in Kitaev's toric code on the plane}
\author{Pieter Naaijkens\footnote{E-mail: \texttt{p.naaijkens@math.ru.nl}}\\%
Institute for Mathematics, Astrophysics and Particle Physics\\ Radboud University Nijmegen, The Netherlands}
\begin{document}
\maketitle
\begin{abstract}
We consider various aspects of Kitaev's toric code model on a plane in the $C^*$-algebraic approach to quantum spin systems on a lattice. In particular, we show that elementary excitations of the ground state can be described by localized endomorphisms of the observable algebra. The structure of these endomorphisms is analyzed in the spirit of the Doplicher-Haag-Roberts program (specifically, through its generalization to infinite regions as considered by Buchholz and Fredenhagen). Most notably, the statistics of excitations can be calculated in this way. The excitations can equivalently be described by the representation theory of $\qdz$, i.e., Drinfel'd's quantum double of the group algebra of $\mathbb{Z}_2$.
\end{abstract}

\section{Introduction}
Kitaev's quantum double model~\cite{MR1951039} has attracted much interest in recent years. One of its interesting features is that the model has anyonic excitations. Such models may be relevant to a new approach to quantum computing, where topological properties of a system are used to do computations (see~\cite{MR2443722,Wang} for reviews). Here we consider the simplest case of this model, corresponding to the group $\mathbb{Z}_2$. This model is often called the \emph{toric code}, although we will consider it on the plane instead of on a torus. This model is not powerful enough for applications to quantum computing, but it has interesting properties nonetheless. In particular, it has anyonic excitations (albeit \emph{abelian} anyons).

The toric code has been studied by many authors by now, for example~\cite{MR1951039,MR2345476,MR1924451}. We take a different viewpoint, namely that of local quantum physics. Indeed, the model can be discussed in the $C^*$-algebraic approach to quantum spin systems~\cite{MR887100,MR1441540}. We show that single excitations can be described by states that cannot be distinguished from the ground state when restricted to measurements outside a cone extending to infinity. This structure is familiar from the algebraic approach to quantum field theory~\cite{MR1405610}, in particular when massive particles are considered~\cite{MR660538}.

The states describing these single excitations lead, via the GNS construction, to inequivalent representations (superselection sectors) of the observable algebra. In fact, these states fulfill a certain selection criterion, pertaining to the fact that they are localized and transportable. The analysis of such representations is central to the Doplicher-Haag-Roberts (DHR) program in algebraic quantum field theory~\cite{MR0297259,MR0334742}. In particular, it turns out that these representations can equivalently be described by endomorphisms of the observable algebra. This description leads in a natural way to the notion of composition of excitations and to statistics of (quasi)particles from first principles. This analysis can be carried out completely for the toric code on the plane. 

A related approach is taken for example in~\cite{MR1463825,MR1234107}, where the authors consider $G$-spin (or, more generally, Hopf-$C^*$) chains. There, excitations localized in \emph{bounded} regions (satisfying the so-called DHR criterion) are considered. Since every injective endomorphism of a finite dimensional algebra is in fact an automorphism, the authors consider \emph{amplimorphisms} to obtain non-abelian charges. Here, we take a different approach, and look instead at endomorphisms localized in certain infinite ``cone'' regions. In our model the irreducible endomorphisms are all automorphisms, but since we consider excitations localized in \emph{infinite} regions, finite dimensionality of the algebras is not an obstruction any more. The idea of construction charged sectors localized in infinite regions is not new: it is used, for example, in the work of Fredenhagen and Marcu~\cite{MR728449}.

Discrete gauge theories in $d=2+1$ show similar algebraic features (i.e., fusion and braiding) of anyons~\cite{MR1163527}. Similar models have been studied in the constructive approach to quantum fields in lattice gauge theory, in particular for the gauge group $\mathbb{Z}_2$ in~\cite{MR728449,MR925923}. These results have been generalized to the group $\mathbb{Z}_N$ in~\cite{MR1341694,MR1604344}. Although the setting considered here is different, some of the methods used are similar. A field theoretic interpretation of the model discussed here can be found in Section 4 of~\cite{MR1951039}.

The paper is organized as follows. In Section~\ref{sec:model}, we recall the model and discuss the ground state in the $C^*$-algebraic setting. In Section~\ref{sec:auto} localized automorphisms describing excitations are described. Section~\ref{sec:fusion} is devoted to fusion and statistics of excitations. Then follows a discussion of operator-algebraic aspects of von Neumann algebras generated by observables localized in cones. Finally, in the last section we prove that the excitations are described by the representation theory of the quantum double $\qd{\mathbb{Z}_2}$.

\section{The model}
\label{sec:model}
We describe Kitaev's model in the $C^*$-algebraic framework for quantum lattice systems~\cite{MR2345476}. Consider a square $\mathbb{Z}^2$ lattice. On each bond of the lattice, i.e. an edge between two vertices of distance 1, there is a spin-1/2 particle. That is, at each bond $b$ the local state space is $\mc{H}_{\set{b}} = \mathbb{C}^2$, with observables $\alg{A}(\{b\}) = M_2(\mathbb{C})$. The set of bonds will be denoted by $\bonds$. If $\Lambda \subset \bonds$ is a finite set, $\alg{A}(\Lambda)$ is the algebra of observables living on the bonds of $\Lambda$. It is the tensor product of the observable algebras acting on the individual bonds of $\Lambda$. If $\Lambda_1 \subset \Lambda_2$ there is an obvious inclusion of corresponding algebras, by identifying $\mc{H}_{\Lambda_2} \cong \mc{H}_{\Lambda_1} \otimes \mc{H}_{\Lambda_2 \setminus \Lambda_1}$. This defines a local net of algebras, with respect to the inclusion $\alg{A}(\Lambda_1) \hookrightarrow \alg{A}(\Lambda_2)$ for $\Lambda_1 \subset \Lambda_2$. Define
\[
	\alg{A}_{loc} = \bigcup_{\Lambda_f \subset \bonds}  \alg{A}(\Lambda_f),
\]
the algebra of local observables. The union is over the finite subsets $\Lambda_f$ of $\bonds$. The algebra $\alg{A}$ of quasi-local observables is the completion of $\alg{A}_{loc}$ in the norm topology, turning it into a $C^*$-algebra. Alternatively, one can see it as the inductive limit of the net $\Lambda \mapsto \alg{A}(\Lambda)$ in the category of $C^*$-algebras. Note that $\alg{A}$ is a uniformly hyperfinite (UHF) algebra~\cite{MR887100}. The algebra of observables localized in an arbitrary subset $\Lambda$ of $\bonds$ is defined as
\[
\alg{A}(\Lambda) = \overline{\bigcup_{\Lambda_f \subset \Lambda} \alg{A}(\Lambda_f)}^{\norm{\cdot}},
\]
where the union is again over finite subsets. An operator $A$ is said to have support in $\Lambda$, or to be localized in $\Lambda$, if $A \in \alg{A}(\Lambda)$. The set $\supp(A) \subset \bonds$ is the smallest subset in which $A$ is localized. 

\begin{figure}
  \begin{center}
  \includegraphics{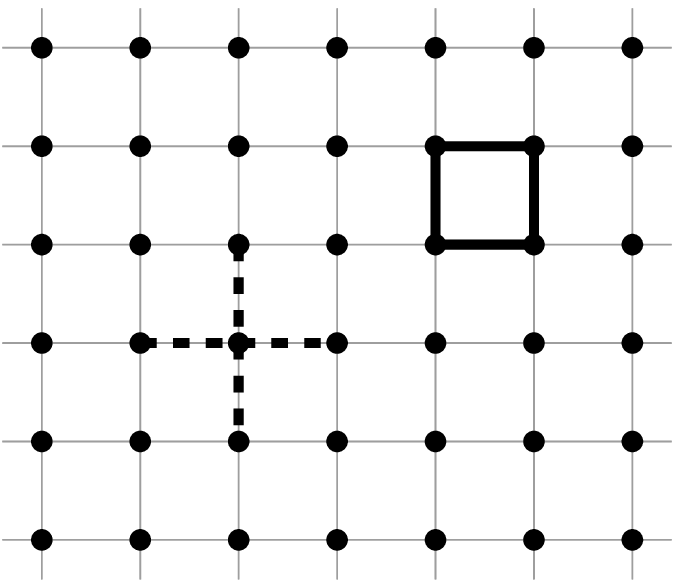}
\end{center}
\caption{The $\mathbb{Z}^2$ lattice. The gray bonds each carry a spin-1/2 particle. A star (dashed lines) and plaquette (thick lines) are shown.}
\label{fig:lattice}
\end{figure}
The Hamiltonian of Kitaev's model is defined in terms of plaquette and star operators, each supported on four bonds (see Figure~\ref{fig:lattice}). If $s$ is a point on the lattice, $\str(s)$ denotes the star based at $s$. Similarly, $\plaq(p)$ are the bonds enclosing a plaquette $p$. The corresponding star and plaquette operators are given by
\[
A_s = \bigotimes_{j \in \str(s)} \sigma_j^x, \qquad B_p = \bigotimes_{j \in \plaq(p)} \sigma_j^z,
\]
where the tensor product is understood as having Pauli matrices $\sigma^x$ (resp. $\sigma^z$) in places $j$, and unit operators in all other positions. It is then straightforward to check that for all stars $s$ and plaquettes $p$, we have
\[
	[A_s, B_p] = 0.
\]
These operators are used to define the local Hamiltonians. If $\Lambda_f \subset \bonds$ finite, the associated local Hamiltonian is
\[
H_{\Lambda_f} = -\sum_{\str(s) \subset \Lambda_f}  A_s -\sum_{\plaq(b) \subset \Lambda_f} B_p.
\]
There is a natural action of $\mathbb{Z}^2$ on the quasi-local algebra, acting by translations. Denote this action by $\tau_x$ for $x \in \mathbb{Z}^2$. Note that the interactions are of finite range, and moreover, they are translation invariant. Hence, there exists an action $\alpha_t$ of $\mathbb{R}$ on $\alg{A}$ describing the dynamics of the system~\cite{MR1441540}, as well as a derivation $\delta$ that is the generator of the dynamics. For observables localized in a finite set $\Lambda$, the action of this derivation is given by\footnote{To be a bit more precise: the derivation $\delta$ defined here is norm-closable and it is the closure $\overline{\delta}$ that generates the dynamics~\cite[Thm. 6.2.4]{MR1441540}. By a density argument, it is often enough to consider $\delta$ instead of its closure.}
\[
\delta(A) = i[H_\Lambda, A], \quad A \in \alg{A}(\Lambda).
\]
By definition, ground states for these dynamics are states $\omega$ of $\alg{A}$ such that $-i \omega(X^* \delta(X)) \geq 0$ for all $X \in \alg{A}_{loc}$.

In~\cite{MR2345476} it is shown that the model admits a unique ground state, which can be computed explicitly. Since we will need the argument later, for the convenience of the reader we summarize the results. The following lemma is crucial in the computation of the ground state. The proof is a straightforward application of the Cauchy-Schwartz inequality and the fact that for $A$ positive, $\omega(A) = 0$ implies that $\omega(A^2) = 0$.
\begin{lemma}
	\label{lem:state}
	Let $\omega$ be a state on a $C^*$-algebra $\alg{A}$, and suppose $X=X^*$ such that $X \leq I$ and $\omega(X) = 1$. Then $\omega(XY) = \omega(YX) = \omega(Y)$ for any $Y \in \alg{A}$.
\end{lemma}

Consider now the abelian algebra $\alg{A}_{XZ}$ generated by the star and plaquette operators. This algebra is in fact \emph{maximal abelian}: $\alg{A}_{XZ}' \cap \alg{A} = \alg{A}_{XZ}$~\cite{MR2345476}. Let $\omega$ be the state on $\alg{A}_{XZ}$ such that $\omega(A_s) = \omega(B_p) = 1$ for all plaquette and star operators.\footnote{That such a state exists can be seen by mapping the model to an Ising spin model.} With help of the lemma, this completely determines the state on $\alg{A}_{XZ}$. Moreover, it minimizes the local Hamiltonians, hence any ground state of the system must be equal to $\omega$ if restricted to $\alg{A}_{XZ}$. The goal is then to show that this state has a unique extension to $\alg{A}$.

Let $\omega_0$ be an extension of $\omega$ to the algebra $\alg{A}$.\footnote{By the Hahn-Banach theorem an extension $\omega_0$ of $\omega$ to $\alg{A}$ always exists.} Using the lemma one can show that for $X,Y \in \alg{A}_{loc}$,
\begin{equation}
  \label{eq:derivation}
  \begin{split}
  -i \omega_0(X^* \delta(Y)) = \sum_{s} &(\omega_0(X^*Y)  - \omega_0(X^* A_s Y))  \\ &+ \sum_{p} (\omega_0(X^*Y) - \omega_0(X^* B_p Y)),
  \end{split}
\end{equation}
where the variable $s$ runs over all stars in the lattice, and $p$ over all plaquettes. If one takes $X = Y$, an application of the Cauchy-Schwartz inequality shows that the right hand side is positive, hence $\omega_0$ is a ground state. 

As mentioned before, in the model at hand this extension is actually unique. In fact, let $X$ be a monomial in the Pauli matrices, say $X = \prod_{i \in \Lambda} \sigma_i^{k_i}$ where $\Lambda \subset \bonds$ is finite and $k_i = x,y$ or $z$. Then $\omega_0(X)$ is non-zero if and only if $X$ is a product of star and plaquette operators, in which case it is $1$. This completely determines the state $\omega_0$, since the value of $\omega_0(X)$ can be computed by a repeated application of Lemma~\ref{lem:state}. For example, to make plausible why $\omega_0$ is zero if $X$ is not a product of star and plaquette operators, consider an operator of the form $A = \sigma^x_j$ for some bond $j$. Then there is a plaquette $p$ such that $j \in \plaq(p)$. But then
\[
	\omega_0(A) = \omega_0(B_p \sigma_j^x B_p) = -\omega_0(A).
\]
In particular, for a local observable $A$ that is a monomial in the Pauli matrices, the set of bonds where $A$ has a $\sigma^x$ component should have the property that the intersection with each plaquette $\plaq(p)$ has an even number of elements. Continuing in this manner, one can show that indeed only products of star and plaquette operators lead to non-zero expectation values~\cite{MR2345476}.

\begin{proposition}
	\label{prop:gstate}
  There is a unique (hence pure) ground state $\omega_0$. This state is translation invariant. The self-adjoint $H_0$ generating the dynamics in the GNS representation $(\pi_0, \mc{H}_0, \Omega)$, when normalized such that $H_0 \Omega = 0$, satisfies $\spec(H_0) \subset \{0\} \cup [4, \infty)$.
\end{proposition}
\begin{proof}
We have already discussed existence and uniqueness of $\omega_0$. Translations map star operators into star operators, and plaquette operators into plaquette operators, hence the ground state is translation invariant.

Since $\omega_0$ is a ground state, it is invariant under the dynamics and the time evolution can be implemented by a strongly continuous group $t \mapsto U_t$ of unitaries. We can choose $U_t$ such that $U_t \Omega = \Omega$. It follows that there is an (unbounded) self-adjoint $H_0$ such that $U_t = e^{it H_0}$ and $H_0 \Omega = 0$.

We claim that $\spec{H_0} \subset \{0\} \cup [M, \infty)$ is equivalent to
\begin{equation}
  \label{eq:spectrum}
  -i\omega_0(X^* \delta(X)) \geq M \left(\omega_0(X^*X) - |\omega_0(X)|^2\right),
\end{equation}
for all $X \in \alg{A}_{loc}$, because the ground state is non-degenerate. Indeed, since $H_0 \Omega = 0$ with $\Omega$ the GNS vector, the inequality can equivalently be written as $\langle X \Omega, H_0 X \Omega \rangle \geq M(\norm{X \Omega}^2 - |\langle \Omega, X \Omega\rangle|^2)$ because $\langle X \Omega , H_0 X \Omega \rangle = \omega_0(X^* \delta(X))$. Here we have identified $X$ with its image $\pi_0(X)$, which is possible since $\pi_0$ is a representation of a UHF (hence simple) algebra. On the other hand, the spectrum condition is equivalent to $H_0 + M P_\Omega \geq M I$, where $P_\Omega$ is the projection on the subspace spanned by $\Omega$ (by non-degeneracy, this is the spectral projection corresponding to $\{0\}$). This is equivalent to the condition
\[
	\langle \Psi, (H_0+MP_\Omega) \Psi \rangle = \langle \Psi, H_0 \Psi \rangle + M |\langle \Omega, \Psi \rangle|^2 \geq M \norm{\Psi}^2
\]
for all $\Psi$ in the domain $D(H_0)$ of $H_0$. But $\pi(\alg{A}_{loc}) \Omega$ is a core for $H_0$ (compare with the proof of~\cite[Prop. 5.3.19]{MR1441540}), hence it is enough to check the inequality for $\Psi = X \Omega$ with $X \in \alg{A}_{loc}$. This shows that inequality~\eqref{eq:spectrum} is equivalent to the assertion on the spectrum of $H_0$.

We now show that inequality~\eqref{eq:spectrum} indeed holds for $M=4$. As a first step, we claim that if either $X$ or $Y$ is a local operator in $\alg{A}_{XZ}$,
\begin{equation}
	-i\omega_0(X^* \delta(Y)) = 4 \left(\omega_0(X^*Y) - \overline{\omega_0(X)} \omega_0(Y) \right) = 0.
\end{equation}
Under these assumptions, the left-hand side can be seen to vanish by equation~\eqref{eq:derivation} and Lemma~\ref{lem:state}. As for the right hand side, consider the case where $X \in \alg{A}_{XZ}$ (the other case is proved similarly). In this case, $X = \sum_{i} \lambda_i X_i$ where each $X_i$ is a product of star and plaquette operators. Using Lemma~\ref{lem:state} again, it follows that $\omega_0(X^*Y) = \sum_i \overline{\lambda_i} \omega_0(Y) = \overline{\omega_0(X)} \omega_0(Y)$, proving the claim.

Now consider the general case, with a local operator $X = X_{XZ} + \sum_{i \in I} \lambda_i X_i$, where $X_0 \in \alg{A}_{XZ}$ and each $X_i$ (with $i$ in some finite set $I$) is a monomial in the Pauli matrices such that $X_i \notin \alg{A}_{XZ}$. Since $X_i \notin \alg{A}_{XZ}$, there is some $A_s$ or $B_p$ that does not commute with $X_i$. Suppose this is $A_s$. Since $X_i$ is a monomial in the Pauli matrices, this actually implies that $\{A_s, X_i\} = 0$, in other words, they anti-commute. Note that this implies that $\omega_0(X_i)$ is zero for each $i \neq 0$, since by the same trick as used before it follows that $\omega_0(X_i) = - \omega_0(X_i)$. By the remarks above, equation~\eqref{eq:spectrum} reduces to
\begin{equation}
	\label{eq:derivac}
	-i \sum_{i,j \in I} \omega_0(X_i^* \delta(X_j)) \geq 4 \sum_{i,j \in I} \omega_0(X^*_i X_j).
\end{equation}
Note that for each $X_i$, there is a finite number $n_i$ of plaquette and star operators that anti-commute with $X_i$. In fact, $n_i \geq 2$, since if there is for example one star operator that does not commute with $X_i$, there must necessarily be another one with this property.\footnote{This amounts to saying that excitations always exist in pairs in finite regions in Kitaev's model~\cite{MR1951039}.} Note that if $n_i \neq n_j$, there is a star or a plaquette operator that commutes with $X_i$ and anti-commutes with $X_j$ (or vice versa). Consequently, $\omega_0(X_i^*X_j) = 0$. 

Now define for each integer $k$ the finite set $I_k = \{ i \in I : n_i = k\}$ and the operators $\widetilde{X}_k = \sum_{i \in I_k} X_i$, with the understanding that $\widetilde{X}_k = 0$ if $I_k$ is the empty set. By the considerations above, it then follows that $\sum_{i,j \in I} \omega_0(X_i^* X_j) = \sum_{k \geq 2} \omega_0(\widetilde{X}_k^* \widetilde{X}_k)$, since $n_i \geq 2$ for each $i \in I$. On the other hand, from equation~\eqref{eq:derivation} it follows that $-i \omega_0(X_i^* \delta(X_j)) = 2 n_i \omega(X_i^* X_j)$. It then follows that the left hand side of the inequality~\eqref{eq:derivac} is equal to $2 \sum_{k \geq 2} k \omega_0(\widetilde{X}_k^* \widetilde{X}_k)$. From this it easily follows that inequality~\eqref{eq:derivac} holds.
\end{proof}

The spectrum condition has far-reaching consequences for the correlation functions; for example, it implies that ground state correlations decay exponentially~\cite{MR2217299}.

\section{Localized endomorphisms}
\label{sec:auto}
In this section we describe localized excitations of the system. In his model, Kitaev associates certain string operators to paths on the lattice (or the dual lattice). These string operators create excitations at the endpoints of the paths~\cite{MR1951039}. The idea is to consider a \emph{single} excitation by moving one of the excitations to infinity, as is done for example in Ref.~\cite{MR728449}. Before this construction is introduced, we give some preliminary definitions.

\begin{figure}
  \begin{center}
  \includegraphics{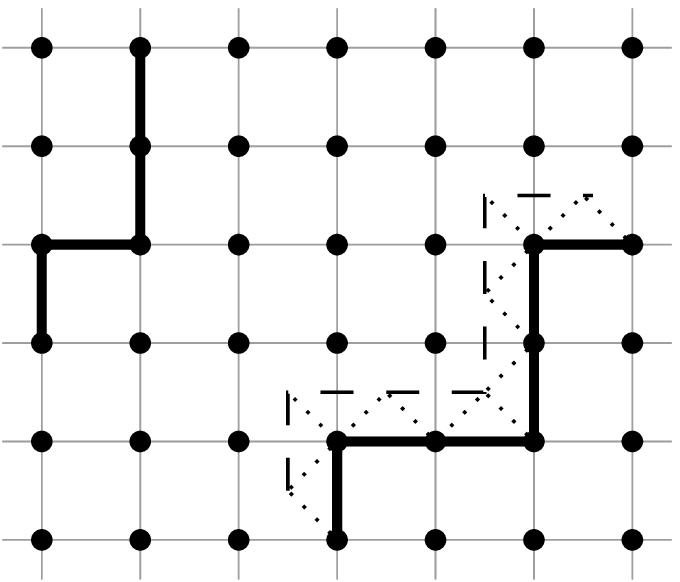}
\end{center}
\caption{A path on the lattice (left black line) and a ribbon. The dots on the ribbon indicate a combined site, i.e. a plaquette with one of its vertices.}
\label{fig:ribbon}
\end{figure}
By a \emph{site}, we mean either a point on the lattice, a plaquette, or a pair of a plaquette with one of its vertices (i.e., a \emph{combined site}). Sites can be seen as the places where excitations can be introduced. Between two sites of the same type, we can consider paths. A path between two points on the lattice is just a path consisting of bonds of the lattice. A path between plaquettes can be viewed as a path on the dual lattice. A path between combined sites is called a \emph{ribbon} (see Figure~\ref{fig:ribbon}). One can think of a ribbon as being composed by a path on the lattice and one on the dual lattice.

\begin{definition}
  Let $\gamma$  be a finite path between two sites. If $\gamma$ is a path on the lattice, define the corresponding \emph{string operator} as $\Gamma_Z^\gamma = \bigotimes_{i \in \gamma} \sigma^z_i$. If it is a path on the dual lattice, the string operator is defined as $\Gamma_X^\gamma = \bigotimes_{i \in \gamma} \sigma^x_i$. Here $i \in \gamma$ means that $i$ is a bond that intersects the path on the dual lattice. Finally, a string operator corresponding to a ribbon is a combination of these constructions. That is, $\Gamma^\gamma_Y = \Gamma^{\gamma_1}_X \Gamma^{\gamma_2}_Z$, where $\gamma_1$ is the path on the lattice and $\gamma_2$ the path on the dual lattice, corresponding to the ribbon.
\end{definition}
It should be clear from the context whether we consider paths on the lattice, paths on the dual lattice, or ribbons. We say that a path or the corresponding string operator is of type X,Y or Z, corresponding to the subscripts used in the definition. 

We first make some observations that will be used later. Consider a plaquette $p$. The corresponding plaquette operator $B_p$ is just the string operator $\Gamma^\gamma_Z$, where $\gamma$ is the closed path consisting of the edges of the plaquette. If $p'$ is, for example, a plaquette adjacent to $p$, $B_p B_{p'}$ is the string operator corresponding to the closed path on the outer edges of the two plaquettes. Continuing this way, it follows that the string operator corresponding to a closed path on the lattice is the product of plaquette operators corresponding to the plaquettes enclosed by the path. The reader will have no trouble checking that similarly a string operator corresponding to a closed path on the dual lattice is the product of all star operators corresponding to the stars enclosed by the path.

The idea now is to study ``elementary'' excitations by first considering a pair of excitations (created by a string operator), and then move one of the excitations to infinity. This technique is also used in, for instance, lattice gauge theory~\cite{MR1604344,MR728449}. We show that in Kitaev's model such excitations can be described by localized automorphisms of $\alg{A}$.
\begin{definition}
Let $\rho$ be a $*$-endomorphism of $\alg{A}$. Let $\Lambda \subset \bonds$ be arbitrary. Then $\rho$ is said to be \emph{localized in $\Lambda$} if $\rho(A) = A$ for all $A \in \alg{A}(\Lambda^c)$. Here $\Lambda^c$ denotes the complement of any subset $\Lambda$ of $\bonds$. 
\end{definition}

\begin{figure}
  \begin{center}
  \includegraphics{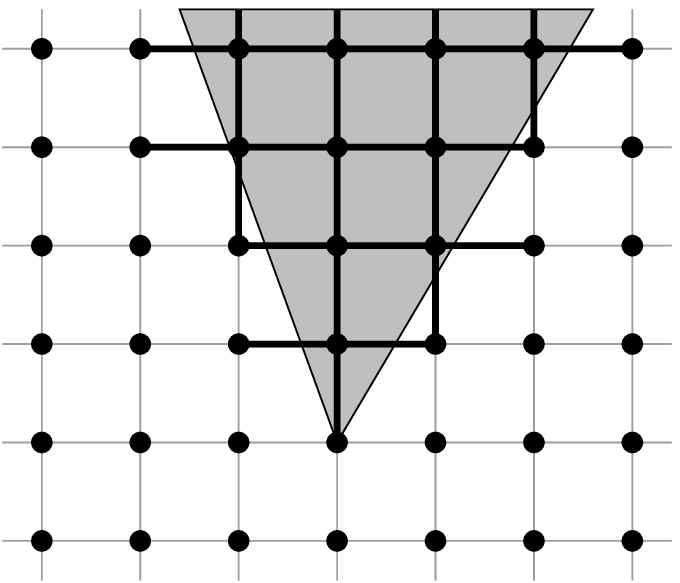}
\end{center}
\caption{Example of a cone (bold bonds). The shaded region is the area bounded by two lines emanating from a point.}
\label{fig:cone}
\end{figure}

We will primarily be interested in cone regions, although in fact the specific shape of the regions is not important (see also Remark~\ref{rem:cone} below).
\begin{definition}
	Consider a point on the lattice $\mathbb{Z}^2$, with two semi-infinite lines emanating from it, such that the angle between those lines is positive but smaller than $\pi$. A \emph{cone} $\Lambda \subset \bonds$ consists of all bonds that are in the area bounded by the two lines, or intersected by one of the lines. See Figure~\ref{fig:cone} for an example.
\end{definition}
Remark that for $x \in \mathbb{Z}^2$ there is a translated cone $\Lambda + x$. Furthermore, $\bigcup_{x \in \mathbb{Z}^2} (\Lambda + x)$ is the set of all bonds. Finally, $\tau_x(\alg{A}(\Lambda)) = \alg{A}(\Lambda+x)$ for any $\Lambda \subset \bonds$. These properties hold in fact for any subset $\Lambda$ of the bonds. 

The string operators induce localized endomorphisms (in fact, automorphisms) of $\alg{A}$. If $\gamma$ is a path starting at a site $x$ and extending to infinity, write $\gamma_n$ ($n \in \mathbb{N}$) for the finite path consisting of the first $n$ bonds of the path $\gamma$.
\begin{proposition}
	\label{prop:auto}
	Let $\Lambda$ be a cone and let $k=X,Y,Z$. Choose a path $\gamma^k$ of type $k$ in $\Lambda$ extending to infinity. Consider the corresponding string operators $\Gamma^{\gamma_n}_k$ for $n \in \mathbb{N}$. For any $A$ in $\alg{A}$, define
\begin{equation}
  \rho^k(A) = \lim_{n \to \infty} \Ad \Gamma_k^{\gamma_n}(A),
	\label{eq:automorphism}
\end{equation}
where the limit is taken in norm. Then for each $k$, $\rho^k$ defines an outer automorphism of the quasi-local algebra $\alg{A}$. These automorphisms are localized in $\Lambda$.
\end{proposition}
\begin{proof}
In the proof we will omit the symbol $\gamma$ and write $\Gamma^n_k$. Suppose $A$ is an observable localized in a finite region $\Lambda_0$. Then one can find $n_0$ such that $(\gamma_n\setminus~\gamma_{n_0}) \cap \Lambda_0 = \emptyset$ for all $n > n_0$. In other words, new parts of the path all lie outside $\Lambda_0$. But then it follows that $\Ad \Gamma_k^n(A) = \Ad \Gamma_k^{n_0}(A)$ for all $n > n_0$, hence the limit in equation~\eqref{eq:automorphism} converges in norm for any local operator $A$.

	To define $\rho^k$ on $\alg{A}$, extend by continuity. Indeed, since each $\Gamma^k_n$ is a unitary operator, $\norm{\rho^k(A)} = \norm{A}$ for each local observable. The local observables are norm-dense in $\alg{A}$, so that  $\rho^k$ extends uniquely to $\alg{A}$. By continuity of the $*$-operation and joint continuity of multiplication (in the norm topology), $\rho^k$ is a $*$-endomorphism. The localization property immediately follows from locality: if $B \in \alg{A}(\Lambda^c)$, then it commutes with $\Gamma_k^n$ for each $n$.

	The endomorphism $\rho^k$ is in fact an automorphism. Indeed, because Pauli matrices square to the identity, $\rho^k \circ \rho^k$ is the identity. To see that the automorphisms are outer, it is enough to notice that the sequence $\Gamma_k^n$ is not a Cauchy sequence in $\alg{A}$, hence it does not converge to an element in $\alg{A}$. By Theorem 6.3 of~\cite{MR1642584}, it follows that the automorphisms are outer.\footnote{Alternatively, this follows because the GNS representation of $\omega_0 \circ \rho^k$ is disjoint from the GNS representation of $\omega_0$, see Theorem~\ref{thm:select}.}
\end{proof}
Note that the automorphism $\rho^k$ depends on the choice of path $\gamma^k$. If necessary, this path dependence will be emphasized by using the notation $\rho^k_\gamma$.

The automorphisms defined in Proposition~\ref{prop:auto} induce states by composing with the ground state. 
\begin{definition}
Let $x$ be a site and $\gamma$ a path of type $k=X,Y,Z$ starting at $x$ and extending to infinity. Define a state $\omega_k^x$ of $\alg{A}$ by $\omega_k^x(A) = \omega_0(\rho^k_\gamma(A))$.
\end{definition}
At first sight, this state appears to depend on the specific choice of path. However, this is not the case.
\begin{lemma}
For each $k=X,Y,Z$ and each site $x$ of the same type, the state $\omega_k^x$ only depends on $x$, but not on the path $\gamma$.
\end{lemma}
\begin{proof}
First consider the case $k=Z$, so that $x$ is a point on the lattice. To prove independence of the path, consider another point $y$ and let $\gamma^1$ and $\gamma^2$ be two paths from $x$ to $y$. Denote the corresponding string operators by $\Gamma_Z^1$ and $\Gamma_Z^2$. This allows to define two (a priori distinct) states 
\[
\omega^{x,y}_i(A) = \omega_0(\Gamma_Z^i A \Gamma_Z^i), \quad i=1,2.
\]
Note that the string operators commute with plaquette operators, hence clearly $\omega^{x,y}_i(B_p) = 1$ for each plaquette $p$. As for the star operators, note that each star has an even number (0,2 or 4) of edges in common with the paths $\gamma^i$, except at the endpoints $x$ and $y$, where there are an odd number of edges in common. Let $s$ be the star based at $x$. Suppose for the sake of example that it has one edge in common with the path $\gamma^1$. Then, using the commutation relations for Pauli matrices,
\[
	\omega_1^{x,y}(A_s) = \omega_0(\Gamma^1_Z A_s \Gamma^1_Z) = i^2 \omega_0(A_s) = -1.
\]
A similar calculation holds in the case of 3 common edges, or for a star $s$ containing the endpoint $y$. Summarizing, we find that $\omega_{1}^{x,y}$ and $\omega_2^{x,y}$ coincide on the abelian algebra $\alg{A}_{XZ}$, taking the value $1$ on all plaquette operators. On the star operators they take the value $-1$ if the star is based at either $x$ or $y$, and $1$ otherwise. A similar argument as given for the ground state now allows us to compute the value of the states on arbitrary elements of the local algebras, and it follows that both states coincide. 

There is in fact another way to see this. Let for example $\gamma$ be a finite path of type $Z$. Let $p$ be a plaquette such that $p \cap \gamma$ is non-empty. Then it is easy to see that $\Gamma_\gamma^Z B_p = \Gamma_{\gamma'}^Z$, where the path $\gamma'$ is obtained from $\gamma$ by deleting the bonds of $\gamma \cap p$ and adding the bonds $p \setminus \gamma$ to the path $\gamma$. Hence once can use the plaquette operators to deform one path into another, provided the endpoints are the same. Since
\[
\omega_0(\Gamma^{Z}_\gamma A \Gamma^{Z}_\gamma) = \omega_0(B_p \Gamma^{Z}_\gamma A \Gamma^{Z}_\gamma B_p)  = \omega_0(\Gamma^{Z}_{\gamma'} A \Gamma^{Z}_{\gamma'})
\]
it follows that the states coincide. A similar argument can be given for paths of type $X$.

Now consider the case where $\gamma^1$ and $\gamma^2$ are two paths starting at $x$ and extending to infinity. Let $A$ be a local observable, localized in some finite set $\Lambda \subset \bonds$. Then there is an $n_0$ such that the paths $\gamma^1_n$ and $\gamma^2_n$ do not return to $\Lambda$ for $n \geq n_0$. Consider a path $\gamma' \subset \Lambda^c$ from $\gamma^1_{n_0}$ to $\gamma^2_{n_0}$. By locality and the result above, we then have
\[
\begin{split}
\omega_0(\rho^Z_{\gamma^1}(A)) = \omega_0(\Gamma_Z^{\gamma^1_{n_0}} A \Gamma_Z^{\gamma^1_{n_0}}) = \omega_0(\Gamma_Z^{\gamma'} \Gamma_Z^{\gamma^1_{n_0}} A \Gamma_Z^{\gamma^1_{n_0}}\Gamma_Z^{\gamma'}) \\
	= \omega_0(\Gamma_Z^{\gamma^2_{n_0}} A \Gamma_Z^{\gamma^2_{n_0}}) = \omega_0(\rho^Z_{\gamma^2}(A)).
\end{split}
\]
By continuity this result extends to observables $A \in \alg{A}$, hence the state $\omega_Z^x$ is independent of the path.

The argument for the states $\omega_X^x$ and $\omega_Y^x$ is essentially the same. The difference is that one has to consider points $x,y$ in the dual lattices, i.e. plaquettes of the lattice, together with paths on the dual lattice. E.g., for $k=X$ one finds
\[
\omega_X^{x,y}(A_s) = 1,\quad \omega_X^{x,y}(B_p) = \begin{cases}
  -1 & x,y \in p \\
  1 & \textrm{otherwise}.
\end{cases}
\]
The argument is now the same as for $\omega^x_Z$.
\end{proof}

The state $\omega^x_k$ describes a single excitation. By the GNS construction, this leads to a corresponding representation $\pi_{\omega^x_k}$ of $\alg{A}$. The GNS triple coming from the ground state $\omega_0$ will be denoted by $(\pi_0, \mc{H}_0, \Omega)$. The remarkable feature is that representations corresponding to single excitations cannot be distinguished from the ground state representation when restricted to the complement of a cone.
\begin{theorem}
  \label{thm:select}
Let $\Lambda \subset \bonds$ be any cone. Then
\begin{equation}
  \label{eq:selection}
  \pi_0 \upharpoonright \alg{A}(\Lambda^c) \cong \pi_{\omega_k^x} \upharpoonright \alg{A}(\Lambda^c),
\end{equation}
for $k=X,Y,Z$ and any site $x$. In addition, $\pi_{\omega_k^x} \cong \pi_{\omega_l^y}$ if and only if $k = l$. This holds for $k=0,X,Y,Z$, where $\omega_0^x := \omega_0$.
\end{theorem}
\begin{proof}
  Let $x$ be a site. Choose a path $\gamma$ (of type $k$) in $\Lambda$, starting at $x$ and going to infinity. Consider $\rho^k := \rho^k_\gamma$ as above. Then $\pi_0 \circ \rho^k$ is localized in $\Lambda$, in the sense that $\pi_0 \circ \rho^k(A) = \pi_0(A)$ for all $A \in \alg{A}(\Lambda^c)$. Moreover, it is a GNS representation for the state $\omega_k^x$, essentially by definition of $\omega_k^x$ (the Hilbert space is $\mc{H}_0$ and $\Omega$ the cyclic vector).\footnote{Note that $\omega_0$ and $\omega_x^k$ are automorphic states in the terminology of~\cite[Ch. 12]{MR1468230}. The statement is then an example of Proposition 12.3.3 of the same reference.} Hence by uniqueness of the GNS representation, $\pi_0 \circ \rho^k \cong \pi_{\omega^x_k}$. Together with localization this yields equation~\eqref{eq:selection}.
	
	Let $y$ be another site. Consider a path $\gamma'$ from $x$ to $y$, with corresponding string operator $\Gamma_k^{\gamma'}$. Note that $\Ad \Gamma_k^{\gamma'} \circ \rho^k$ is precisely the automorphism induced by the path from $y$ to infinity, obtained by concatenating $\gamma'$ with $\gamma$. From unitarity of $\Gamma_k^{\gamma'}$ it is easy to see that $\pi_{\omega_k^x} \cong \pi_{\omega_k^y}$, proving that the GNS representations of type $k$ are equivalent, independent of the starting site.
\begin{figure}
  \begin{center}
  \includegraphics{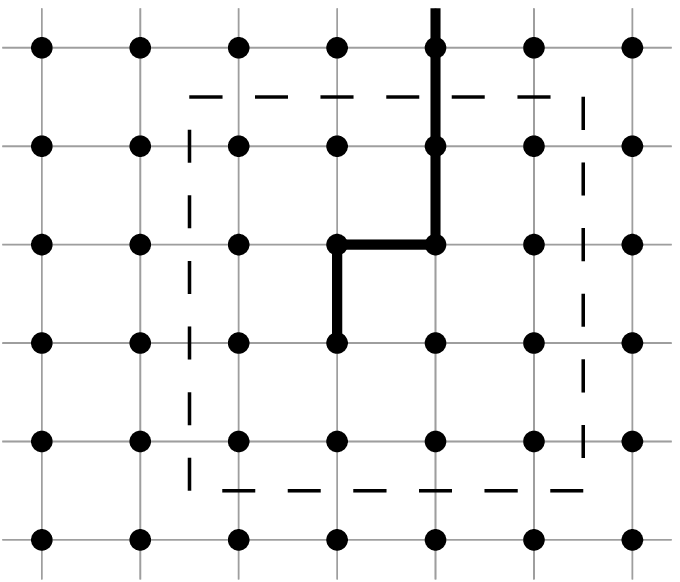}
\end{center}
\caption{Consider the state induced by thick path on the lattice. A path $\gamma$ on the dual lattice (dashed) defines a string operator $\Gamma_X^\gamma$. The state has value $-1$ on this operator.}
\label{fig:inequiv}
\end{figure}

To complete the proof, we show that the representations are globally inequivalent. Note that $\omega_0$ is a pure state, hence the GNS representation is irreducible. The GNS representations of the states $\omega_k$ can be obtained by composing $\pi_0$ with an automorphism of $\alg{A}$, hence they are also irreducible. But this implies that $\omega_0$ and $\omega_k$ are factor states. Moreover, since the representations are irreducible, unitary equivalence is equivalent to quasi-equivalence of the states~\cite[Prop. 10.3.7]{MR1468230}. Recall that in the situation at hand, two factor states $\omega_1$ and $\omega_2$ are quasi-equivalent if and only if for each $\varepsilon > 0$, there is a finite set of bonds $\widehat{\Lambda}$ such that for all finite sets $\widetilde{\Lambda} \subset \widehat{\Lambda}^c$ and $B \in \alg{A}(\widetilde{\Lambda})$, $|\omega_1(B)-\omega_2(B)| < \varepsilon \norm{B}$, by Corollary 2.6.11 of~\cite{MR887100}. We show that this inequality cannot hold.

Consider for the sake of example the case $\omega_0$ and $\omega_Z^x$, for some point $x$ on the lattice. Set $\varepsilon = 1$. Without loss of generality, we can assume that $\widehat{\Lambda}$ contains the star based at $x$. Since $\widehat{\Lambda}$ is finite, it is possible to choose a closed non-self-intersecting path $\gamma$ in the dual lattice, such that the set $\widehat{\Lambda}$ is contained in the region bounded by the path (see Figure~\ref{fig:inequiv}). Consider the string operator $\Gamma_X^\gamma$ corresponding to this path. Then clearly this operator is localized in a finite region in the complement of $\widehat{\Lambda}$. Recall that $\Gamma_X^\gamma$ is the product of star operators enclosed by the path $\gamma$, in particular the star based at $x$. That is, $\Gamma_X^\gamma = A_{\str(x)} A_{s_1} \cdots A_{s_n}$ for certain stars $s_1, \dots s_n$. But this implies
\[
|\omega_0(\Gamma^\gamma_X) - \omega_Z^x(\Gamma^\gamma_X)| = |1-\omega_Z^x(A_{\str(x)})| = 2 > \norm{\Gamma^\gamma_X}.
\]
The other cases are similar, if necessary using plaquettes instead of stars.
\end{proof}
\begin{remark}\label{rem:cone}
	The fact that $\Lambda$ is a cone is not essential at this point. What is important is that it should be possible to choose a path extending to infinity contained in $\Lambda$. In particular, the proof implies that it is not possible to sharpen the result to unitary equivalence when restricted to the complement of a \emph{finite} set. At one point in the analysis however, notably in the proof of Theorem~\ref{thm:type}, it is essential to be able to translate the support of  any local observable to a region completely inside $\Lambda$. If $\Lambda$ is a cone, this is always possible.
\end{remark}

In the language of algebraic quantum field theory, the representations $\pi_{\omega_k}$ are said to satisfy a \emph{selection criterion}. Usually one imposes such a selection criterion to select physically relevant representations. Here however, we start with physically reasonable constructions and arrive at the criterion. The criterion here can be interpreted as a lattice analogue of localization in spacelike cones, as considered in~\cite{MR660538}. An example of a model admitting such representations, albeit a model mainly of mathematical interest, is constructed in~\cite{BF-ICMP1981}. The interpretation is that the excitations cannot be distinguished from the ground state outside a cone region. It would be interesting to know if there are other irreducible representations of $\alg{A}$, not unitarily equivalent to the representations in Theorem~\ref{thm:select}, satisfying this criterion. One probably has to impose additional criteria to select physically relevant representations (cf. the condition on the existence of a mass gap in~\cite{MR660538}).

For the automorphisms considered here a similar property can be derived. In particular, the automorphisms are covariant with respect to the time evolution. Moreover the generator has positive spectrum bounded away from zero. Note that the algebra $\alg{A}$ (being UHF) is simple, hence $\pi_0$ is a faithful representation. To simplify notation, from now on we identify $\pi_0(A)$ with $A$ and often drop the symbol $\pi_0$, as already done in the proof of Proposition~\ref{prop:gstate}.
\begin{proposition}
	Let $\gamma$ be a path to infinity of type $k$. Then $\rho_\gamma$ is covariant for the action of $\alpha_t$. In fact, suppose $\gamma$ is of type $Z$. Then, for all $t \in \mathbb{R}$ and $A \in \alg{A}$,
\[
		\rho_\gamma(\alpha_t(A)) = e^{it (H_0+2 A_s)} \rho_\gamma(A) e^{-it(H_0+2 A_s)}
\]
with $\spec(H_0+ 2 A_s) \subset [2,\infty)$. Here $s$ is the starting point of $\gamma$. For the case $k=X$ one has to replace $A_s$ by $B_p$, where $p$ is the plaquette where the path starts. The case $k = Y$ has generator $H_0 + 2B_p + 2A_s$, with spectrum contained in $[4,\infty)$.
\end{proposition}
\begin{proof}
	We prove the result for paths of type $X$. The other cases are proved by making the obvious modifications. First note that for $A \in \alg{A}_{loc}$, $\alpha_t(A) = \lim_{\Lambda \to \mathbb{Z}^2} e^{i H_\Lambda t} A e^{-i H_\Lambda t}$ in norm.

By the same reasoning as in the proof of Lemma~\ref{lem:state}, one sees that $\rho_\gamma(A_s) = -A_s$. Hence if $\Lambda \supset \str(s)$, we have $\rho_\gamma(H_\Lambda) = H_\Lambda + 2 A_s$. By expanding the exponential into a power series, it is then clear that
\[
	\rho_\gamma(e^{i t H_\Lambda} A e^{-i t H_\Lambda}) = e^{it(H_\Lambda+ 2 A_s)} \rho_\gamma(A)  e^{-it(H_\Lambda+ 2 A_s)}.
\]
One then sees (remark in particular that $A_s$ commutes with all local Hamiltonians) that for all $A \in \alg{A}$ we have $\rho_\gamma(\alpha_t(A)) = U_t \rho_\gamma(A) U_t^*$, where $U_t$ is the unitary $U_t = \exp(i t (H_0+2 A_s))$.

It remains to show the spectrum condition. This can be done by similar methods as used in the proof of Proposition~\ref{prop:gstate}. The spectrum condition is equivalent to the inequality
\[
	-i \omega(X^*\delta(X)) + 2 \omega(X^* A_s X) - 2 \omega(X^*X) \geq 0
\]
for all $X \in \alg{A}_{loc}$. We then proceed as before: write $X = X_{XZ} + \sum_i X_i$ where $X_{XZ} \in \alg{A}_{XZ}$ and $X_i \notin \alg{A}_{XZ}$ monomials in the Pauli matrices. After substituting this into the inequality, all terms containing $X_{XZ}$ vanish. By the same reasoning as in the proof of Proposition~\ref{prop:gstate} one then sees that this inequality is indeed satisfied for all $X \in \alg{A}_{loc}$.
\end{proof}

The following corollary is immediate.
\begin{corollary}
	The states $\omega_{\gamma}^k$ are invariant with respect to $\alpha_t$.
\end{corollary}

\section{Fusion, statistics and braiding}
\label{sec:fusion}
The localized endomorphisms considered in the previous section can be endowed with a tensor product. In fact, it is possible to define a braiding in a canonical way. This braiding is related to the statistics of particles. In the DHR analysis, a crucial role in the construction is played by Haag duality in the vacuum sector. For dealing with cone localized endomorphisms, the appropriate formulation is the condition that for each cone $\Lambda$ the following equality holds:
\[
\pi_0(\alg{A}(\Lambda))'' = \pi_0(\alg{A}(\Lambda^c))'.
\]
Note that by locality, one always has $\pi_0(\alg{A}(\Lambda))'' \subset \pi_0(\alg{A}(\Lambda^c))'$. Currently no general conditions from which Haag duality follows are known, but note that there are some results for quantum spin chains, e.g.~\cite{MR2281418,MR2605849}.  At the moment we do not have a proof of Haag duality, but since the ground state is known explicitly, one might hope that a direct proof is possible.

Fortunately, in the present situation it is possible to do without Haag duality. To clarify this, first note that Theorem~\ref{thm:select} implies in particular that the localized automorphisms defined by paths extending to infinity are \emph{transportable}.
\begin{definition}
	Let $\Lambda$ be a cone and suppose that $\rho$ is an endomorphism of $\alg{A}$ localized in $\Lambda$. Then $\rho$ is called \emph{transportable}, if for any cone $\widehat{\Lambda}$ there is a unitary equivalent\footnote{We do not require that this unitary lives in $\alg{A}$. More precisely, we demand that $\pi_0 \circ \rho \cong \pi_0 \circ \widehat{\rho}$.} endomorphism $\widehat{\rho}$ localized in $\widehat{\Lambda}$.
\end{definition}
One of the applications of Haag duality is to get more control over the unitary setting up the equivalence. Specifically, one can show that the intertwiners are elements of the (weak closure) of cone algebras. Recall that an intertwiner $V$ from an endomorphism $\rho_1$ to $\rho_2$ is an operator such that $V \rho_1(A) = \rho_2(A) V$ for all $A \in \alg{A}$. A unitary intertwiner is also called a charge transportation operator (or simply charge transporter). In our model we will be able to prove, without invoking Haag duality, that the charge transporters are elements of the weak closure of cone algebras. We again identify $\pi_0(A)$ with $A$ in the proof.

\begin{lemma}
  \label{lem:intertwiner}
Let $\gamma^1$ (resp. $\gamma^2$) be a path of type $k$ starting at a site $x$ (resp. $y$) and extending to infinity. Then there is a unitary intertwiner $V$ from $\rho_{\gamma_1}^k$ to $\rho_{\gamma_2}^k$ such that $\Gamma_{\widehat{\gamma}}^k V \Omega = \Omega$ (where $\Omega$ is the GNS vector for $\omega_0$) for any path $\widehat{\gamma}$ from $x$ to $y$.

Moreover, if for each $n$ a path $\widetilde{\gamma}_n$ from the $n$-th site of $\gamma_1$ to the $n$-th site of $\gamma_2$ is chosen such that $\lim_{n \to \infty} \dist(\widetilde{\gamma}_n, x) = \infty$, then for $V_n = \Gamma_1^n \Gamma^k_{\widetilde{\gamma}_{n}}\Gamma_2^n$, where $\Gamma_i^n$ is the string operator corresponding to the path $\gamma^i_n$, we have
\begin{equation}
V = \wlim_{n \to \infty} V_n.
\end{equation}
In other words, $V_n$ is a sequence of operators converging weakly to $V$.
\end{lemma}
\begin{proof}
First note that $A_s \Omega = B_p \Omega = \Omega$ for all star and plaquette operators. Indeed,
\[
\norm{(A_s - I) \Omega}^2 = \omega_0( (A_s-I)^*(A_s-I) ) = 2-2 = 0,
\]
for all $A_s$. A similar calculation holds for the operators $B_p$. Note that this property can be interpreted as the ground state vector minimizing the value of each local Hamiltonian~\cite{MR1951039}.

First note that a unitary $V$ as in the statement is necessarily unique because any unitary intertwiner from $\rho_{\gamma_1}^k$ to $\rho_{\gamma_2}^k$ is a scalar multiple of $V$, by Schur's lemma and irreducibility of $\pi_0$. To show existence, first consider (for simplicity) the case where $\gamma^1$ and $\gamma^2$ start at the same site $x$. As remarked earlier in the proof of Theorem~\ref{thm:select}, $\Omega$ is a cyclic vector for $\rho_1^k$ and for $\rho_2^k$ (we will write $\rho_1^k$ instead of $\rho_{\gamma^1}^k$ in the proof). Moreover, the corresponding vector state is $\omega^x_k$. By uniqueness of the GNS construction, there is a unitary $V$ such that $V \rho_1^k(A) = \rho_2^k(A) V$ for all $A \in \alg{A}$, and $V\Omega = \Omega$. 

Choose paths $\widetilde{\gamma}_n$ as in the statement of the lemma. The path obtained by concatenating $\widetilde{\gamma}_n$ with the paths $\gamma^1_n$ and $\gamma^2_n$ can be seen as a loop based at $x$ that gets larger and larger as $n$ gets bigger. Now consider a sequence $V_n$ of unitaries defined by $V_n = \Gamma_1^n \Gamma_{2}^n \Gamma_{\widetilde{\gamma}_n}^k$ where $\Gamma_i^n$ is defined in the statement of the Lemma. Note that $V_n$ is a product of star and plaquette operators, since it is the path operator of a closed loop. Hence, $V_n \Omega = \Omega$ by the observation above. Suppose $B \in \alg{A}_{loc}$. Let $N$ be such that $\widetilde{\gamma}_n \cap \supp(B) = \emptyset$ for all $n \geq N$. Then from locality, one can easily verify that $V_n \rho_1^k(B) = \rho^k_2(B) V_n$ for all $n \geq N$, in other words,
\[
	\lim_{n \to \infty} \langle \rho_1^k(A) \Omega, V_n \rho_1^k(B) \Omega \rangle = \lim_{n \to \infty} \langle \rho_1^k(A) \Omega, \rho_2^k(B) V_n\Omega \rangle =  \langle \Omega, \rho_1^k(A)^* \rho^k_2(B) \Omega \rangle,
\]
for all $A,B \in \alg{A}_{loc}$. On the other hand, for each $A,B \in \alg{A}_{loc}$,
\[
	\langle \rho_1^k(A) \Omega, V \rho_1^k(B) \Omega \rangle = \langle \Omega, \rho_1^k(A)^* \rho_2^k(B) \Omega \rangle,
\]
since $V \Omega = \Omega$. The sequence $V_n$ is uniformly bounded and because $\rho_1^k(\alg{A}_{loc}) \Omega$ is dense in $\mc{H}_0$, since $\rho_1^k$ is an automorphism, it follows that $V_n \to V$ weakly. Seeing that any path $\widehat{\gamma}$  from $x$ to $x$ is a loop, it is clear that $\Gamma_{\widehat{\gamma}}^k V \Omega = \Omega$. 

As for the general case, suppose $\gamma^1$ starts at the site $x$ and $\gamma^2$ starts at the site $y$. Choose a path $\widetilde{\gamma}$ from $x$ to $y$. Then $\widehat{\rho} := \Ad \Gamma^k_{\widetilde{\gamma}} \circ \rho_1^k$ is defined by a path starting at $y$. By the argument above, there is a unitary $\widehat{V}$ intertwining $\widehat{\rho}$ and $\rho_2^k$ such that $\widehat{V} \Omega = \Omega$. Set $V = \Gamma_{\widetilde{\gamma}}^k \widehat{V}$. It follows that $V$ is an intertwiner from $\rho_1^k$ to $\rho_2^k$ that satisfies $\Gamma_{\overline{\gamma}}^k V \Omega = \Omega$ for all paths $\overline{\gamma}$ from $x$ to $y$, because $\Gamma_{\widetilde{\gamma}}^k \Gamma_{\overline{\gamma}}^k$ is the path operator of a loop. The claim on the converging net follows from the construction.
\end{proof}

A pleasant consequence of the above proof is that an explicit sequence converging to the intertwiners is given, which makes it possible to do explicit calculations. A direct consequence of the Lemma is that we have some control over the algebras containing the unitary intertwiners, a point where usually Haag duality is used.
\begin{theorem}
  \label{thm:intertwiner}
  Suppose $\Lambda_1$ and $\Lambda_2$ are two cones such that there is another cone $\Lambda \supset \Lambda_1 \cup \Lambda_2$. For $k=X,Y,Z$, consider $\rho_i^k \cong \pi_{\omega_k}$ localized in $\Lambda_i$ for $i=1,2$, defined by paths $\gamma^i$ extending to infinity. Let $W$ be a unitary such that $W \rho_1^k(A) = \rho_2^k(A) W$ for all $A \in \alg{A}$. Then $W \in \alg{A}(\Lambda)''$.
\end{theorem}
\begin{proof}
  By Schur's lemma, $W$ is a multiple of the intertwiner $V$ in the previous lemma. The geometric situation makes it clear that a net $V_n$ as in the lemma can be chosen to be a net in $\alg{A}(\Lambda)$. This net converges weakly to $V$, by the previous Lemma.
\end{proof}
\begin{remark}
	Again it is not essential that $\Lambda$ as in the theorem is a cone. It is enough to be able to chose paths $\widetilde{\gamma}_n$ in as in Lemma~\ref{lem:intertwiner} that lie inside $\Lambda$. But note that the smaller $\Lambda$ is, the more control one has over the algebra where the intertwiners live in.
\end{remark}

\begin{proposition}
  The representations $\rho^k$ are covariant with respect to the action $\tau_x$ of translations. That is, for each $x \in \mathbb{Z}^2$ there is a unitary $W(x)$ such that $\rho^k(\tau_x(A)) = W(x) \rho^k(A) W(x)^*$ for all $A \in \alg{A}$ and the map $x \mapsto W(x)$ is a group homomorphism.
\end{proposition}
\begin{proof}
	Let $\gamma$ denote the string (starting at the site $x_0$) defining $\rho^k$. For $x \in \mathbb{Z}^2$, consider the translated string $\widehat{\gamma} = \gamma-x$. This defines an automorphism $\widehat{\rho}^k$. In fact, $\widehat{\rho}^k = \tau_{-x} \circ \rho^k \circ \tau_x$. Then by Lemma~\ref{lem:intertwiner} there is a unitary intertwiner $V_x$ from $\rho^k$ to $\widehat{\rho}^k$. We choose $V_x$ such that the condition in Lemma~\ref{lem:intertwiner} is satisfied. 
  
Write $U(x)$ for the unitaries that implement the translations in the GNS representation of $\omega_0$. Define $W(x) = U(x) V_x$. It then follows that $\rho^k(\tau_x(A)) = W(x) \rho^k(A) W(x)^*$ for all $A \in \alg{A}_{loc}$, and hence by continuity for all $A \in \alg{A}$. It remains to show that $W(x)$ is a representation of $\mathbb{Z}^2$. By irreducibility of $\rho^k$ it follows that $W(x+y) = \lambda(x,y) W(x) W(y)$ with $\lambda$ a 2-cocycle of $\mathbb{Z}^2$ taking values in the unit circle. The claim is that $\lambda$ is in fact trivial.
  
This would follow from the equation $U(y)^* V_x U(y) = V_{x+y} V_y^*$ for all $x,y \in \mathbb{Z}^2$. Note that the operator on the right hand side is an intertwiner from $\rho_{\gamma-y}^k$ to $\rho_{\gamma-(x+y)}^k$ satisfying the condition in Lemma~\ref{lem:intertwiner}. This equation can be verified by noting that $V_{x+y}$ and $V_y$ commute with path operators (this should be clear from the construction of a converging net) and by the following observation: a path operator $\Gamma_{\widehat{\gamma}}$ (where $\widehat{\gamma}$ is a path from $x_0-y$ to $x_0-(x+y)$) can be written as $\Gamma_{\gamma_1} \Gamma_{\gamma_2}^*$ with $\gamma_1$ a path from $x_0$ to $x_0-(x+y)$ and $\gamma_2$ a path from $x_0$ to $x_0-y$. Let $V^n_x$ be a sequence as in Lemma~\ref{lem:intertwiner} converging weakly to $V_x$. Then for the translated sequence $\tau_{-y}(V^n_x)$
\[
\wlim_{n \to \infty} \tau_{-y}(V^n_x) = V_{x+y} V_y^*,
\]
by the same Lemma. The result follows since the map $A \mapsto \tau_{-y}(A) = U(y)^* A U(y)$ is weakly continuous, hence the left hand side is equal to $U(y)^* V_x U(y)$.
\end{proof}
It is possible to define a tensor product of localized endomorphisms. If $\rho_1$ and $\rho_2$ are localized in cones $\Lambda_1$ and $\Lambda_2$, the basic idea is to define an endomorphism $\rho_1 \otimes \rho_2$ by $(\rho_1 \otimes \rho_2)(A) = \rho_1(\rho_2(A))$. If $\Lambda \supset \Lambda_1 \cup \Lambda_2$ is a cone, it follows that $\rho_1 \otimes \rho_2$ is localized in $\Lambda$. In order to get a categorical tensor structure, one would then like to define a tensor product for \emph{intertwiners}. If $T_i$, $i=1,2$ are intertwiners from $\rho_i$ to $\sigma_i$ (and $T_i \in \alg{A}$), the reader will have no difficulty showing that $T_1 \otimes T_2 := T_1 \rho_1(T_2)$ is an intertwiner from $\rho_1 \otimes \rho_2$ to $\sigma_1 \otimes \sigma_2$. In the terminology of category theory, this would turn the category of cone localized automorphisms with intertwiners as morphisms into a strict tensor category. The trivial endomorphism $\iota$ is the tensor unit. Note that the unit operator $I$ of $\alg{A}$ can be regarded as an intertwiner from $\rho$ to itself for any endomorphism $\rho$. To indicate this, we sometimes write $I_\rho$. The distinction is important in the definition of the tensor product of intertwiners.

There is, however, one problem with this definition: the intertwiners are elements of the algebra $\alg{A}(\Lambda)''$ rather than of $\alg{A}(\Lambda)$ (recall that we identified $\pi_0(\alg{A})$ with $\alg{A}$). There is no reason why they should be contained in the quasi-local algebra $\alg{A}$, because this algebra is not weakly closed in general. Since the localized endomorphisms are (a priori) only defined on $\alg{A}$, the above definition therefore does not make sense. 

A possible solution is to introduce an \emph{auxiliary algebra} that contains the intertwiners~\cite{MR660538}. Choose an arbitrary cone $\Lambda_a$, which will be fixed from now on. The cone can be interpreted as a ``forbidden'' direction, not unlike the technique of puncturing the circle. Introduce a partial ordering on $\mathbb{Z}^2$ by defining
\[
	x \leq y \Leftrightarrow (\Lambda_a +y) \subset (\Lambda_a+x) \Leftrightarrow (\Lambda_a+x)^c \subset (\Lambda_a+y)^c.
\]
Now $(\mathbb{Z}^2, \leq)$ is a directed set (each pair of points has an upper bound with respect to $\leq$), hence it is possible to take the ($C^*$)-inductive limit
\begin{equation}
	\alg{A}^{\Lambda_a} = \overline{\bigcup_{x \in \mathbb{Z}^2} \alg{A}((\Lambda_a+x)^c)''}^{\norm{\cdot}}.
  \label{eq:auxalg}
\end{equation}
Note that $\alg{A}^{\Lambda_a+x} = \alg{A}^{\Lambda_a}$ for all $x \in \mathbb{Z}^2$. Clearly, $\alg{A} \subset \alg{A}^{\Lambda_a}$. Moreover, if $\Lambda$ is a cone such that $\Lambda \subset (\Lambda_a+x)^c$ for some $x$, then $\alg{A}(\Lambda)'' \subset \alg{A}^{\Lambda_a}$. An important point\footnote{\label{ftn:endo}In the case of algebraic quantum field theory, the main point is to obtain \emph{endomorphisms} of the auxiliary algebra from \emph{representations} of the quasi-local algebra. In the present model, however, we already have automorphisms of $\alg{A}$.} is that the automorphisms we consider can be extended to $\alg{A}^{\Lambda_a}$.

\begin{proposition}
	Let $\rho$ be an automorphism defined by a path extending to infinity. Then $\rho$ has a unique extension $\rho^{\Lambda_a}$ to $\alg{A}^{\Lambda_a}$ that is weakly continuous on $\alg{A}( (\Lambda_a+x)^c)''$ for any $x \in \mathbb{Z}^2$. Moreover, $\rho^{\Lambda_a}(\alg{A}^{\Lambda_a}) \subset \alg{A}^{\Lambda_a}$; in other words, it is an endomorphism of the auxiliary algebra.
\end{proposition}
\begin{proof}
	The proof is essentially the same as that of Lemma 4.1 of~\cite{MR660538}, except at points where duality is used. First, let $A \in \alg{A}( (\Lambda_a+x)^c)$. Since $\rho$ is localizable, there is a unitary $V$ such that $\rho(A) = V A V^*$ (choose a unitary equivalent endomorphism localized in $\Lambda_a+x$). This implies that $\rho$ is weakly continuous on $\alg{A}( (\Lambda_a+x)^c )$ and the unique weakly continuous extension can be given by $\rho^{\Lambda_a}(B) = V B V^*$ for $B \in \alg{A}( (\Lambda_a+x)^c)''$. This procedure determines $\rho^{\Lambda_a}$ on all of $\alg{A}^{\Lambda_a}$.

	To show that $\rho^{\Lambda_a}$ maps $\alg{A}^{\Lambda_a}$ into itself, first note that $\rho(\alg{A}(\Lambda)) \subset \alg{A}(\Lambda)$ for every finite set $\Lambda \subset \bonds$. Hence, by weak continuity,
\[
\rho^{\Lambda_a}(\alg{A}( (\Lambda_a+x)^c)'') = \rho(\alg{A}( (\Lambda_a+x)^c))'' \subset \alg{A}( (\Lambda_a+x)^c)'',
\]
which proves the claim.
\end{proof}
\begin{remark}
  In the proof of Buchholz and Fredenhagen, Haag duality is used to show that the extensions map the auxiliary algebra into itself (see also Footnote~\ref{ftn:endo}). The point is that using Haag duality it is possible to show that for representations localized in a cone $\Lambda$ one has $\rho(\alg{A}(\Lambda)) \subset \alg{A}(\Lambda)''$. Since we have an explicit description of the representations, we can directly prove the stronger statement $\rho(\alg{A}(\Lambda)) \subset \alg{A}(\Lambda)$ for the automorphisms considered in our model. However, the intertwiners are typically \emph{not} elements of $\alg{A}(\Lambda)$. 
\end{remark}

We now redefine the tensor product as $\rho_1 \otimes \rho_2 = \rho_1^{\Lambda_a} \circ \rho_2$. For the automorphisms that we have considered so far, this definition reduces to the old one. However, to define the tensor product of intertwiners, this definition is necessary. If $T$ is an intertwiner from $\rho_1$ to $\rho_2$ and $T \in \alg{A}(\Lambda)''$ for some cone $\Lambda$ asymptotically disjoint from $\Lambda_a$, then $S \otimes T := S \rho_1^{\Lambda_a}(T)$ is a well-defined intertwiner from $\rho_1 \otimes \rho_2$ to $\rho_1' \otimes \rho_2'$.

The tensor product gives rise to \emph{fusion rules}. A fusion rule gives a decomposition of the tensor product of two irreducible representations into a direct sum of irreducible representations. In Kitaev's model the rules are particularly simple. As remarked before, for each $k=X,Y,Z$, $\rho^k \otimes \rho^k = \iota$, where $\iota$ is the trivial endomorphism of $\alg{A}$. Furthermore, essentially by definition, $\rho^X \otimes \rho^Z \cong \rho^Y$. This determines the fusion rules for unitarily equivalent representations as well: unitaries setting up the equivalence can be defined using the tensor product.

Using the tensor product, in this case a \emph{braiding} can then be defined, similarly as in the DHR analysis~\cite{MR0297259}. This is a unitary operator $\varepsilon_{\rho_1, \rho_2}$ intertwining $\rho_1 \otimes \rho_2$ and $\rho_2 \otimes \rho_1$. First, consider two disjoint cones $\Lambda_1$ and $\Lambda_2$ that are both contained in $(\Lambda_a+x)^c$ for some $x$. We say that $\Lambda_1 < \Lambda_2$ if we can rotate $\Lambda_1$ counter-clockwise around the apex of the cone until it has non-empty intersection with $\Lambda_a+x$, such that at any intermediate angle it is disjoint from $\Lambda_2$. Note that for two disjoint cones either $\Lambda_1 < \Lambda_2$ or $\Lambda_2 < \Lambda_1$.

Now let $\rho_1, \rho_2$ be two localized automorphisms, as considered above, such that $\rho_1$ is localized in a cone $\Lambda_1$ and $\rho_2$ in $\Lambda_2$. Moreover, we demand that there is a cone $\Lambda \supset \Lambda_1 \cup \Lambda_2$. Note that $\rho_1 \otimes \rho_2$ is localized in $\Lambda$. Choose a cone $\widehat{\Lambda}_2$ such that $\widehat{\Lambda}_2 < \Lambda_1$. Then there is a unitary $V$ such that $V \rho_2(-) V^*$ is localized in $\widehat{\Lambda}_2$. This unitary can be chosen in $\alg{A}^{\Lambda_a}$~\cite{Naaijkens}. It then follows that $\varepsilon_{\rho_1, \rho_2} := (V \otimes I_{\rho_1})^*(I_{\rho_1} \otimes V) = V^* \rho_1^{\Lambda_a}(V)$ is an intertwiner from $\rho_1 \otimes \rho_2$ to $\rho_2 \otimes \rho_1$.

With this definition, one can prove the following result by adapting the proof in the DHR analysis (see e.g.~\cite{halvapp}) in a suitable way.
\begin{lemma}
  The braiding $\varepsilon_{\rho,\sigma}$ only depends on the condition $\widehat{\Lambda}_2 < \Lambda_1$, not on the specific choices made. Moreover, it satisfies the \emph{braid equations}
  \begin{equation}
	\begin{split}
	  \varepsilon_{\rho, \sigma \otimes \tau} &=  (I_\sigma \otimes \varepsilon_{\rho,\tau})(\varepsilon_{\rho,\sigma} \otimes I_\tau) \\
		\varepsilon_{\rho \otimes \sigma, \tau} &= (\varepsilon_{\rho,\tau} \otimes I_\sigma)(I_\rho \otimes \varepsilon_{\sigma,\tau}).
	\end{split}
	  \label{eq:braid}
  \end{equation}
  Furthermore, $\varepsilon_{\rho,\sigma}$ is natural in $\rho$ and $\sigma$: if $T$ is an intertwiner from $\rho$ to $\rho'$, then $\varepsilon_{\rho',\sigma} (T \otimes I) = (I \otimes T) \varepsilon_{\rho,\sigma}$, and similarly for $\sigma$.
\end{lemma}

In Lemma~\ref{lem:intertwiner}, a net converging to the charge transporters was explicitly constructed. This makes it possible to calculate the braiding operators exactly. In the subscript of the braiding, we will sometimes write $X,Y$ or $Z$ instead of $\rho^X, \rho^Y$ and $\rho^Z$. 
\begin{theorem}
  \label{thm:statistics}
Let $\rho_1, \rho_2$ be automorphisms defined by strings extending to infinity in some cone $\Lambda$. Suppose that each automorphism is of type X or type Z. The braid operators in each of the possible cases are then given by $\varepsilon_{X,X} = \varepsilon_{Z,Z} = I$ and $\varepsilon_{X,Z} = \pm I$. If $\varepsilon_{X,Z} =I$, then $\varepsilon_{Z,X} = -I$ and vice versa. 
\end{theorem}
\begin{proof}
  Consider a cone $\widehat{\Lambda}$ disjoint from $\Lambda$, such that $\widehat{\Lambda} < \Lambda$ and such that there is a cone $\widetilde{\Lambda} \supset \Lambda \cup \widehat{\Lambda}$. There is a path $\widehat{\gamma}_2$ in $\widetilde{\Lambda}$ such that the corresponding automorphism $\widehat{\rho}_2$ is unitarily equivalent to $\rho_2$ and localized in $\widehat{\Lambda}$. The corresponding unitary charge transporter $V$ is then contained in $\alg{A}(\widehat{\Lambda})''$. By definition we then have $\varepsilon_{\rho_1, \rho_2} = V^* \rho_1^{\Lambda_a}(V)$.

\begin{figure}
  \begin{center}
	\includegraphics{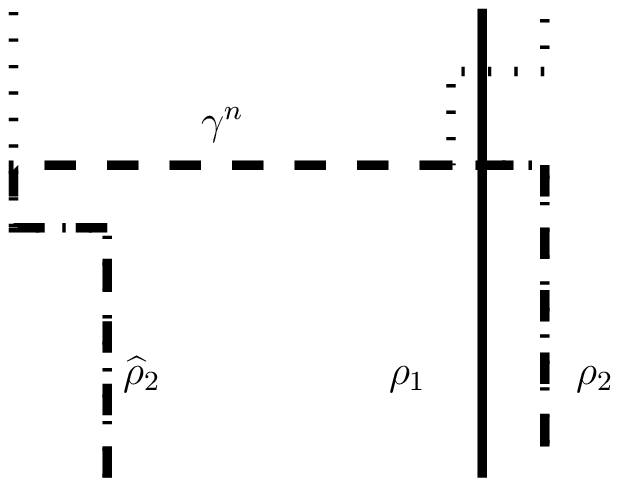}
\end{center}
\caption{The path $\gamma_n$ (dashed line) crosses the defining path of $\rho_1$ from the right. The dotted lines represent the defining paths of $\rho_2$ and $\widehat{\rho}_2$.}
\label{fig:braiding}
\end{figure}

This can be calculated using weak continuity of $\rho_1^{\Lambda_a}$ and the explicit construction of Lemma~\ref{lem:intertwiner} of a net converging to $V$. Indeed, let $V_n \to V$ be this net. Note that each $V_n$ is a string operator of the same type as $\rho_2$. In particular, if $\rho_1$ is of the same type as $\rho_2$, then $\rho_1(V_n) = V_n$ for all $n$ and hence $\rho^{\Lambda_a}_1(V) = V$. It follows that $\varepsilon_{X,X} = \varepsilon_{Z,Z} = I$.

  The situation where $\rho_1$ is of type X and $\rho_2$ is of type Z (or vice versa) is a bit more complicated. Recall that for the definition of the net $V_n$, for each $n$ a path $\gamma^n$ is chosen, such that the distance to the starting points of the paths $\gamma_1$ and $\gamma_2$ goes to infinity. The operator $V_n$ is then the string operator corresponding to the string formed by the first $n$ bonds of $\gamma_2$ and $\widehat{\gamma}_2$, together with $\gamma^n$. Note that, if $n$ is big enough, this string crosses $\gamma_1$ either an even number of times, or an odd number, independent of $n$. This property depends on whether the first crossing is from the ``left'' or from the ``right'' (see Figure~\ref{fig:braiding}), or if there is no crossing at all.
  
By anti-commutation of the Pauli matrices, it follows that if the number of crossings is even, $\rho_1(V) = V$, whereas if it is odd then $\rho_1(V) = -V$. Hence, $\varepsilon_{X,Z} = \pm I$. If the role of $\rho_1$ and $\rho_2$ is reversed, an odd number of crossings becomes an even number. This observation proves the last claim.
\end{proof}

Since $\rho^Y = \rho^X \otimes \rho^Z$, the braid equations allow to compute the braiding with excitations of type $Y$. The braiding with the trivial automorphism is always trivial. This completely determines the braiding for all irreducible representations we consider. 

We note that the sign of, for example, $\varepsilon_{X,Z}$ depends on the relative localization of both strings. Indeed, suppose we have two automorphisms $\rho_1$ and $\rho_2$, defined by strings $\gamma_1$ of type $X$ and $\gamma_2$ of type $Z$, extending to infinity and localized in $\Lambda_1$ resp. $\Lambda_2$. Suppose moreover that $\Lambda_2 < \Lambda_1$. It then follows that $\varepsilon_{\rho_1, \rho_2} = I$, since the paths in the proof, going from $\gamma_2$ to $\widehat{\gamma}_2$, do not cross $\gamma_1$. On the other hand, if $\Lambda_1 < \Lambda_2$ it follows that $\varepsilon_{\rho_1,\rho_2} = -I$. Note that this coincides with the situation in algebraic quantum field theory in low dimensions~\cite[Sect. 2.2]{MR1199171}.

The final piece of structure is that of \emph{conjugation}. A conjugate can be interpreted as an anti-charge. Formally, a conjugate for an endomorphism $\rho$ is a triple $(\overline{\rho}, R, \overline{R})$ such that $R$ intertwines $\iota$ and $\overline{\rho} \otimes \rho$ and $\overline{R}$ intertwines $\iota$ and $\rho \otimes \overline{\rho}$~\cite{MR1444286}. Here $\iota$ is the trivial endomorphism. The intertwiners $R, \overline{R}$ should satisfy
\[
\overline{R}^* \rho(R) = I, \quad R^* \overline{\rho}(\overline{R}) = I.
\]
A conjugate for an irreducible endomorphism $\rho$ is called \emph{normalized} if $R^*R = \overline{R}^* \overline{R}$ and \emph{standard} if $R^* \overline{\rho}(T) R = \overline{R}^* T \overline{R}$ for every intertwiner $T$ from $\rho$ to itself. If a conjugate exists, one can always find a standard conjugate.

Note that $\rho^k \otimes \rho^k = \iota$ for $k=X,Y,Z$. It follows that in our model the automorphisms we consider have conjugates. These are particularly simple: $\overline{\rho}^k = \rho^k$ and one can choose the unit operators for the intertwiners $R$ and $\overline{R}$. This is trivially a standard conjugate.

With the help of the braiding and conjugates one can define a \emph{twist}. Let $\rho$ be a cone localized endomorphism and $(\overline{\rho}, R, \overline{R})$ be a standard conjugate. The twist $\Theta_\rho \in \End(\rho)$ is then defined by
\[
\Theta_\rho = (\overline{R}^* \otimes \id_\rho) \circ (\id_{\overline{\rho}} \otimes \varepsilon_{\rho,\rho}) \circ (\overline{R} \otimes \id_\rho).
\]
Note that if $\rho$ is irreducible, $\Theta_\rho = \omega_\rho I$ for some phase factor. The (equivalence class of) $\rho$ is called \emph{bosonic} if $\omega_\rho = 1$ and \emph{fermionic} if $\omega_\rho = -1$. Since the conjugates of $\rho^k$, $k=X,Y,Z$ are particularly simple, the following corollary immediately follows from Theorem~\ref{thm:statistics}.
\begin{corollary}
The excitations $X$ and $Z$ are bosonic and $Y$ is fermionic.
\end{corollary}

\section{Cone algebras}
\label{sec:cones}
Let $\Lambda$ be a cone. In this section we consider the von Neumann algebras associated to the observables localized in this cone. More precisely, define $\mc{R}_\Lambda := \pi_0(\alg{A}(\Lambda))''$ and $\mc{R}_{\Lambda^c} := \pi_0(\alg{A}(\Lambda^c))''$. The main result in this section is that $\mc{R}_\Lambda$ is an infinite factor.

\begin{lemma}
  \label{lem:generate}
  With the notation above, $\mc{R}_\Lambda \vee \mc{R}_{\Lambda^c} = \alg{B}(\mc{H}_0)$.
\end{lemma}
\begin{proof}
  Note that for each set $\Lambda \subset \bonds$ one has $\mc{R}_\Lambda = \bigvee_{b \in \Lambda} \pi_0(\alg{A}(\{b\}))$. It follows that $\alg{B}(\mc{H}_0) = \pi_0(\alg{A})'' = \mc{R}_\Lambda \vee \mc{R}_{\Lambda^c}$.
\end{proof}

More can be said about the cone algebras. In fact, they are infinite factors. In other words, $\mc{R}_\Lambda$ is a factor of Type I$_\infty$, Type II$_\infty$ or Type III. The basic idea of the proof, which is adapted from~\cite[Proposition 5.3]{MR2281418}, is to assume that $\mc{R}_\Lambda$ admits a tracial state. It then follows that $\omega_0$ is tracial, which is a contradiction. In fact, Type I$_\infty$ can be ruled out as well.
\begin{theorem}
	\label{thm:type}
  $\mc{R}_\Lambda$ is a factor of Type $II_\infty$ or Type III.
\end{theorem}
\begin{proof}
  To show that $\mc{R}_\Lambda$ is a factor, we argue as in~\cite{MR2281418}. The center is $\mc{Z}(\mc{R}_\Lambda) = \mc{R}_\Lambda \cap \mc{R}_\Lambda'$. By taking commutants, $\mc{Z}(\mc{R}_\Lambda)' = \mc{R}_\Lambda \vee \mc{R}_\Lambda'$. Note that $\mc{R}_{\Lambda^c} \subset \mc{R}_\Lambda'$, hence by Lemma~\ref{lem:generate}, $\mc{Z}(\mc{R}_\Lambda)' = \alg{B}(\mc{H}_0)$.

  Assume that $\mc{R}_\Lambda$ is a finite factor. Then there exists a unique tracial state $\psi$ on $\mc{R}_\Lambda$. This induces a tracial state $\widetilde{\psi} = \psi \circ \pi_0$ on $\alg{A}(\Lambda)$. By Propositions 10.3.12(i) and 10.3.14 of~\cite{MR1468230}, it follows that the state $\widetilde{\psi}$ is factorial and quasi-equivalent to the restriction of $\omega_0$ to $\alg{A}(\Lambda)$.
  
  Let $\varepsilon > 0$. By Corollary 2.6.11 of~\cite{MR887100}, there is a finite set $\widehat{\Lambda} \subset \Lambda$ such that $|\omega_0(A) - \widetilde{\psi}(A)| < \varepsilon \norm{A}$ for all $A \in \alg{A}(\Lambda \setminus \widehat{\Lambda})$. Now, let $k > 0$ be an integer. Consider local observables $A,B$ with localization region contained in $B(0,k)$ (that is, all bonds that can be connected to the origin of $\mathbb{Z}^2$ with a path of length at most $k$) and norm 1. Since $\Lambda$ is a cone and $\widehat{\Lambda}$ is finite, there is an $x \in \mathbb{Z}^2$, such that $\tau_{x}(AB)$ is localized in $\Lambda \setminus \widehat{\Lambda}$. By translation invariance,
 \[
 |\omega_0(AB) - \widetilde{\psi}(\tau_{x}(AB))| = |\omega_0(\tau_x(AB)) - \widetilde{\psi}(\tau_x(AB))| < \varepsilon,
\]
and similarly for $BA$. Hence since $\widetilde{\psi}$ is a trace,
\[
|\omega_0(AB) - \omega_0(BA)| = |\omega_0(AB) - \widetilde{\psi}(\tau_{x}(AB)) - \omega_0(BA) +  \widetilde{\psi}(\tau_{x}(BA)) | < 2 \varepsilon.
\]
Because $k$ and $\varepsilon$ were arbitrary, $\omega_0(AB) = \omega_0(BA)$ for all $A,B \in \alg{A}_{loc}$, which is absurd.

To see that the Type I case can be ruled out, note that $\mc{R}_\Lambda$ is of Type I if and only if $\omega_0$ is quasi-equivalent to $\omega_{0,\Lambda} \otimes \omega_{0,\Lambda^c}$. This can be seen by adapting the proof of~\cite[Prop. 2.2]{MR1828987}. Let $\widehat{\Lambda} \subset \bonds$ be any finite set. Then one can always find a star $s$ in $\widehat{\Lambda}^c$ such that the intersection with both $\Lambda$ and $\Lambda^c$ is not empty. But for this star $s$, one has $\omega_0(A_s) = 1$. On the other hand, $(\omega_{0,\Lambda} \otimes \omega_{0,\Lambda_c})(A_s) = 0$, essentially because $\Lambda \cap s$ is not a star any more. This implies that the states $\omega$ and $\omega_{0,\Lambda} \otimes \omega_{0,\Lambda^c}$ are not equal at infinity. It follows that $\omega_0$ cannot be quasi-equivalent to $\omega_{0,\Lambda} \otimes \omega_{0,\Lambda^c}$.
\end{proof}

We single out a useful consequence of this result.
\begin{corollary}
  \label{cor:isometry}
  Let $\Lambda$ be a cone. Then $\mc{R}_\Lambda$ contains isometries $V_1, V_2$ such that $V_i^* V_j = \delta_{i,j} I$ and $V_1 V_1^* + V_2 V_2^* = I$.
\end{corollary}
\begin{proof}
	By~\cite[Prop. V.1.36]{MR1873025}, there is a projection $P$ such that $P \sim (I-P) \sim I$, where $\sim$ denotes Murray-von Neumann equivalence with respect to $\mc{R}_\Lambda$. Hence, there are isometries $V_1, V_2$ such that $V_1 V_1^* = P$ and $V_2 V_2^* = (I-P)$. These isometries suffice.
\end{proof}

Although we have no proof for Haag duality for cones, we would like to point out an interesting consequence of this duality. For two cones $\Lambda_1 \subset \Lambda_2$, write $\Lambda_1 \ll \Lambda_2$ if any star or plaquette in $\Lambda_1 \cup \Lambda_2^c$ is either contained in $\Lambda_1$ or in $\Lambda_2^c$.
\begin{definition}
We say that $\omega_0$ satisfies the \emph{distal split property for cones} if for any pair of cones $\Lambda_1 \ll \Lambda_2$ there is a Type I factor $\mc{N}$ such that $\mc{R}_{\Lambda_1} \subset \mc{N} \subset \mc{R}_{\Lambda_2}$.
\end{definition}

With the assumption of Haag duality we can then prove the following theorem.
\begin{theorem}
Suppose that $\pi_0$ satisfies Haag duality for cones. Then $\omega_0$ has the distal split property for cones.
\end{theorem}
\begin{proof}
  Let $\Lambda_1 \ll \Lambda_2$ be two cones. Note that it is enough to prove that $\mc{R}_{\Lambda_1} \vee \mc{R}_{\Lambda_2}' \simeq \mc{R}_{\Lambda_1} \otimes \mc{R}_{\Lambda_2}'$, where $\simeq$ denotes that the natural map $A \otimes B' \mapsto AB'$ ($A \in \mc{R}_{\Lambda_1}, B' \in \mc{R}_{\Lambda_2}'$) extends to a normal isomorphism. Indeed, if this is the case, the result follows from Theorem 1 and Corollary 1 of~\cite{MR703083}, since $\mc{R}_{\Lambda_1}$ and $\mc{R}_{\Lambda_2}$ are factors.

  Note that $\omega_0(AB) = \omega_0(A)\omega_0(B)$ if $A \in \alg{A}(\Lambda_1), B \in \alg{A}(\Lambda_2^c)$. Since $\omega_0$ is normal, this result is also valid for $A \in \mc{R}_{\Lambda_1}$ and $B \in \mc{R}_{\Lambda_2^c}$. A result of Takesaki~\cite{MR0100798} then implies that $\mc{R}_{\Lambda_1 \cup \Lambda_2^c} = \mc{R}_{\Lambda_1} \vee \mc{R}_{\Lambda_2^c} \simeq \mc{R}_{\Lambda_1} \otimes \mc{R}_{\Lambda_2^c}$. By Haag duality, $\mc{R}_{\Lambda_2^c} = \mc{R}_{\Lambda_2}'$, which concludes the proof.
\end{proof}
Note that without Haag duality only the existence of a Type I factor $\mc{R}_{\Lambda_1} \subset \mc{N} \subset \pi_0(\alg{A}(\Lambda_2^c))'$ can be concluded. The condition that $\Lambda_1 \ll \Lambda_2$ is needed precisely to avoid the situation at the end of the proof of Theorem~\ref{thm:type}.

\section{Equivalence with $\Rep_f \qd{\mathbb{Z}_2}$}
\label{sec:equivalence}
If $G$ is a finite group, one can form the \emph{quantum double} $\qd{G}$ of the group. The quantum double is a quasi-triangular Hopf algebra (see e.g.~\cite{MR1321145} for an introduction). It is well-known that $\Rep_f \qd{G}$, the category of finite dimensional $\qd{G}$-modules, is a \emph{modular tensor category}~\cite{MR1797619}. In this section we will introduce the category $\Delta(\Lambda)$ of stringlike localized representations and show that it is equivalent to $\Rep_f \qd{\mathbb{Z}_2}$ (as braided tensor $C^*$-categories). This implies that for all practical purposes, the excitations are described by the representation theory of $\qdz$. 

\begin{lemma}
  Let $\rho_1, \rho_2$ be two transportable endomorphisms of $\alg{A}$, localized in a cone $\Lambda$. Then one can define a localized and transportable \emph{direct sum} $\rho_1 \oplus \rho_2$.
\end{lemma}
\begin{proof}
	Let $V_1, V_2 \in \mc{R}_\Lambda$ be isometries as in Corollary~\ref{cor:isometry}. Define $\rho(A) := V_1 \rho_1(A) V_1^* + V_2 \rho_2(A) V_2^*$, for all $A \in \alg{A}$. It follows that $\rho$ is a $*$-representation\footnote{Note that $\rho$ is not necessarily an endomorphism of $\alg{A}$ any more, but rather of $\alg{A}^{\Lambda_a}$. This is however only a minor technicality and is not essential for what follows.} of $\alg{A}$. Since $V_i \in \mc{R}_\Lambda$ and $\mc{R}_{\Lambda^c} \subset \mc{R}_\Lambda'$, it follows that $\rho(A) = A$ for $A \in \alg{A}(\Lambda^c)$, hence $\rho$ is localized in $\Lambda$. 
  To show transportability, let $\widehat{\Lambda}$ be another cone. Pick isometries $W_1, W_2 \in \mc{R}_{\widehat{\Lambda}}$ as in Corollary~\ref{cor:isometry}. Since $\rho_1$ and $\rho_2$ are transportable, there are unitary operators $U_i$ such that $U_i \rho_i(-) U_i^*$ is localized in $\widehat{\Lambda}$. Define $W = W_1 U_1 V_1^* + W_2 U_2 V_2^*$. Then $W W^* = W^* W = I$ and $W \rho(-) W^*$ is localized in $\widehat{\Lambda}$, hence $\rho$ is transportable. This $\rho$, which is unique up to unitary equivalence, will be denoted by $\rho_1 \oplus \rho_2$.  
\end{proof}

We will now introduce the category $\Delta(\Lambda)$. For technical reasons it is convenient to consider only representations localized in a fixed cone $\Lambda$, since in that case clearly all intertwiners are in the algebra $\alg{A}^{\Lambda_a}$. Proceeding in this way, there is no problem in defining the tensor product. It should be noted that the resulting category does not depend on the specific choice of cone $\Lambda$ (see \cite[Prop. 2.11]{Naaijkens} for a proof and for alternative approaches).

The irreducible objects of the category $\Delta(\Lambda)$ are precisely the automorphisms localized in the cone $\Lambda$ that are given by paths extending to infinity. The morphisms are intertwiners from one endomorphism to another. By the Lemma above, finite direct sums can be constructed, turning $\Delta(\Lambda)$ into a category with direct sums. In fact, by construction, each object can be decomposed into irreducibles. It is clear from the construction that the direct sums can be extended to endomorphisms of the auxiliary algebra. Hence the tensor product defined in Section~\ref{sec:fusion} can be defined for all objects. Similarly, a braiding for direct sums can be constructed from Theorem~\ref{thm:statistics}. Conjugates for direct sums can be constructed from conjugates for the irreducible components. Summarizing, freely using terminology from~\cite{halvapp,mmappendix}, we have the following result:
\begin{theorem}
The category $\Delta(\Lambda)$ is a braided tensor $C^*$-category.
\end{theorem}

The category obtained in this way is actually equivalent (as a braided tensor $C^*$-category) to the representation category of $\qd{\mathbb{Z}_2}$ over the field $k = \mathbb{C}$. For the structure of $\Rep_f \qd{\mathbb{Z}_2}$ as a braided tensor $C^*$-category we refer to Ref.~\cite{MR1234107}. A highbrow way of seeing this is to appeal to the classification results of modular tensor categories~\cite{MR2544735}. It is however possible to give an explicit construction of the equivalence. Note that equivalence as \emph{braided} categories is in general stronger than equivalence as tensor categories. Indeed, there are non-isomorphic groups whose representation categories are equivalent as tensor categories but not as \emph{braided} tensor categories~\cite{MR1810480}. On the other hand, every \emph{symmetric} tensor category (satisfying certain additional properties) is the representation category of a compact group (determined up to isomorphism)~\cite{MR1010160}.
\begin{theorem}
	There is a braided equivalence of tensor $C^*$-categories $\Delta(\Lambda) \to \Rep_f \qd{\mathbb{Z}_2}$.
\end{theorem}
\begin{proof}
  Since $\mathbb{Z}_2$ is abelian, the irreducible representations of $\qd{\mathbb{Z}_2}$ are labeled by the elements $e,f$ of $\mathbb{Z}_2$ and $\chi_e, \chi_\sigma$ of the dual group $\widehat{\mathbb{Z}_2}$~\cite{MR1797619,MR1128130}. Here $\chi_e$ and $\chi_\sigma$ denote the trivial and the sign character of $\mathbb{Z}_2$ respectively. Write $V_{g,\chi}$ for the irreducible $\qdz$-module induced by an element $g$ and character $\chi$. We obtain the following list of all irreducible modules of $\qdz$:
\[
\Pi_0 = V_{e, \chi_e} , \quad \Pi_X = V_{f, \chi_e}, \quad \Pi_Y = V_{f, \chi_\sigma}, \quad \Pi_Z = V_{e, \chi_\sigma}.
\]
Recall that using the coproduct of $\qdz$ the tensor product $\Pi_i \otimes \Pi_j$ can be made into a left $\qdz$-module. The tensor product has the same fusion rules as $\Delta(\Lambda)$, e.g. $\Pi_X \otimes \Pi_Y \cong \Pi_Z$ and $\Pi_k \otimes \Pi_0 \cong \Pi_0 \otimes \Pi_k \cong \Pi_k$.

On the side of $\Delta(\Lambda)$, choose paths of type $X,Z$ such that the corresponding automorphisms $\rho^X, \rho^Z$ satisfy $\varepsilon_{X,Z} = -I$. Define $\rho^Y = \rho^X \otimes \rho^Z$, and $\rho^0 = \iota$, the trivial endomorphism. Note that each irreducible representation in $\Delta(\Lambda)$ is unitarily equivalent to one of the $\rho^k$. This suggests to define  a functor $F: \Rep_f \qdz \to \Delta(\Lambda)$ as follows: for irreducible modules, the most natural choice is to set $F(\Pi_k) = \rho^k$ for $k=0,X,Y,Z$. The irreducible modules have dimension one, hence the $\qdz$-linear maps between the irreducible modules are just the scalars. In order for $F$ to be a linear functor, there is essentially only one choice of $F(T)$ for a morphism $T$. Note that $F$ is full and faithful on the Hom-sets of irreducible objects. By construction every irreducible object of $\Delta(\Lambda)$ is isomorphic to an object in the image of $F$.

In fact, $F$ is a braided monoidal functor. By our particular choice of $\rho^X$, $\rho^Y$ and $\rho^Z$, one can choose the natural transformations $F(V \otimes W) \to F(V) \otimes F(W)$, needed for the definition of a monoidal functor, to be identities. To see that $F$ is indeed a \emph{braided} functor, recall that for $\pi_1,\pi_2 \in \Rep_f \qdz$, the braiding $c_{\pi_1,\pi_2}$ is the linear map intertwining $\pi_1 \otimes \pi_2$ and $\pi_2 \otimes \pi_1$ defined by $c_{\pi_1, \pi_2} = \sigma \circ (\pi_1\otimes \pi_2)(R) $. Here $\sigma$ is the canonical flip and $R$ is a universal $R$-matrix for $\qdz$. It is then straightforward to verify that for irreducible modules, $F$ sends the braiding of $\Rep_f \qdz$ to that of $\Delta(\Lambda)$. For example, $c_{\Pi_X, \Pi_Z} = -1$ (where we omit the isomorphism of the underlying vector spaces).

The extension of the functor to direct sums is left to the reader, as is the verification that $F$ preserves all the relevant structures of a braided tensor $C^*$-category. Since the irreducible objects of both categories are in 1-1 correspondence, and the functor $F$ preserves direct sums and braidings, $F$ sets up an equivalence of braided tensor $C^*$-categories. Note, for example, that $F$ is full, faithful and essentially surjective. Indeed, it is tedious but relatively straightforward to define an inverse functor setting up the equivalence.
\end{proof}

\emph{Acknowledgments:} This research is funded by the Netherlands Organization for Scientific Research (NWO) grant no. 613.000.608. I would like to thank M. Fannes for a discussion on the construction of the ground state, P. Fendley for the idea that single excitations can be obtained by moving one excitation of a pair to infinity and M. M\"uger and N.P. Landsman for helpful discussions and a critical reading of the manuscript. Professors D. Buchholz and K. Fredenhagen gave useful references at the 27$^{th}$ LQP Workshop in Leipzig, where this work was presented. An anonymous referee pointed out a gap in the proof of the spectral gap in an earlier version, as well as a suggestion on how to fix it.

\bibliographystyle{abbrv}
\bibliography{refs}

\end{document}